\documentclass[runningheads]{llncs}
\usepackage[T1]{fontenc}
\usepackage{graphicx}
\usepackage{amsmath}
\usepackage{fancybox}
\usepackage{xspace}
\usepackage{amsfonts}
\usepackage{makecell}
\usepackage{tikz}
\usetikzlibrary{arrows,decorations,decorations.text,decorations.markings}
\usetikzlibrary{shapes,trees,positioning}

\usepackage{hyperref}
\hypersetup{colorlinks=true, unicode=true, linkcolor=[rgb]{0.10,0.05,0.67}, citecolor=[rgb]{0.10,0.05,0.67}, filecolor=[rgb]{0.10,0.05,0.67}, urlcolor=[rgb]{0.10,0.05,0.67}}

\usepackage{orcidlink}

\newcommand{\reservedWordTiles}[1]{\ovalbox{\ensuremath{\mathsf{#1}}\xspace}}
\newcommand{\tiles}[1]{\reservedWordTiles{#1}\xspace}

\newcommand{\Soda}{\textsc{Soda}\xspace}
\newcommand{\Tiles}{\ensuremath{\mathsf{Tiles}}\xspace}

\newcommand{\sBoolean}{\ensuremath{\mathbb{B}}\xspace}
\newcommand{\sAction}{\ensuremath{\mathbf{A}}\xspace}
\newcommand{\sFluent}{\ensuremath{\mathbf{F}}\xspace}

\newcommand{\subseteqFV}{\ensuremath{\subseteq _{\mathsf{F}} \mathsf{FV}}}

\newtheorem{proposition-appendix}{Proposition}[section]

\graphicspath{{images/}}

\begin{document}
    \title{Composable Verification Pipelines for Multi-Agent Systems}
%
%\titlerunning{Abbreviated paper title}
% If the paper title is too long for the running head, you can set
% an abbreviated paper title here
%
    \author{Julian~Alfredo~Mendez\orcidlink{0000-0002-7383-0529} \and
    Andreas Br\"{a}nnstr\"{o}m\orcidlink{0000-0001-9379-4281}}
    \authorrunning{Mendez and Br\"{a}nnstr\"{o}m}
% First names are abbreviated in the running head.
% If there are more than two authors, 'et al.' is used.
%
    \institute{Ume{\aa} University, Sweden \\
    \email{julian.mendez@cs.umu.se, andreas.brannstrom@umu.se}
    }
    \maketitle              % typeset the header of the contribution

    \begin{abstract}
        Existing approaches for reasoning about action and change provide expressive semantics for modeling dynamic systems, in most cases built on top of logic programming systems. We introduce a modular framework for transition and trajectory verification based on \Tiles and implemented in \Soda, which is an efficient functional programming language. The framework operationalizes action language semantics through executable verification pipelines that process states, actions, transitions, and rules as compositional functional components. Verification procedures are represented as typed functional pipelines, enabling modular specifications, reusable reasoning components, and transparent execution workflows with guaranteed pipeline termination. The framework includes an executable specification layer that allows users to define domain descriptions in YAML, which are operationalized into the underlying verification model and executable pipeline structure. We provide an open-source implementation and illustrate the framework through examples that involve misinformation and emotional reasoning.

    \end{abstract}

    \section{Introduction}
    \label{sec:intro}

    The increasing complexity of dynamic and multi-agent systems has intensified the need for executable verification mechanisms capable of analyzing evolving executions under rule-based constraints~\cite{ferrando2025vitamin,engelmann2023rv4jaca}. Modern systems increasingly operate through adaptive interactions, sequential decision making, and evolving behavioral contexts, requiring verification approaches that can transparently represent and analyze system evolution.
    Ensuring that executions satisfy desired properties and avoid undesirable behaviors requires methods that can explicitly represent, operationalize, and verify system evolution~\cite{brannstrom2026humanemotionverificationaction,park2024ai,monteith2024artificial}.

    Seminal work on action languages provide expressive semantics for reasoning about action and change~\cite{dworschak2008mbn,gelfond1998action,brannstrom2026humanemotionverificationaction}.
    Existing approaches to modeling dynamic systems include logic-based \cite{adam2009logical,lorini2011logic,dastani2012logic,balduccini2010formalization,brannstrom2026humanemotionverificationaction,dworschak2008mbn,gelfond1998action}, planning-based \cite{lorini2022cognitive,lorini2022cognitive,davila2021simple,cirillo2010human,McDermott-1998-PDDL}, including PDDL~\cite{McDermott-1998-PDDL}, and a range of agent-based approaches \cite{bolander2011epistemic,korevcko2013some,jones2009personality,pereira2007formal,steunebrink2012formal,sanchez2019designing,rudra2025composable}.
    Implementations found in the literature are usually based on logic programming systems, such as
    GOLOG~\cite{Levesque-1997-Golog}, % - https://en.wikipedia.org/wiki/GOLOG
    STRIPS~\cite{Fikes-1971-Strips}, and % - https://en.wikipedia.org/wiki/Stanford_Research_Institute_Problem_Solver
    DLV$^{\text{K}}$~\cite{Eiter-2003-DVLK}. % - https://www.dbai.tuwien.ac.at/proj/dlv/K/
    To bridge to other communities, such as those in functional and object-oriented programming, modular and composable software design patterns~\cite{harmelen2021modular,mishra2024scalable,rudra2025composable} can be included in these approaches.
    In that manner, this would simplify the deployment of planning and verification solutions with massively used technologies.

    In this paper, we propose a modular framework for transition and trajectory verification based on a compositional representation of functions and their execution. The framework structures verification procedures by composing simple functional units into pipelines. To support this, we employ \Tiles\cite{Mendez.Kampik.Aler.Dignum-2024-SCAI}, a framework for composing typed functions into graphical and modular specifications~\cite{Mendez.Kampik-2025-LNGAI}, together with \Soda\cite{Mendez-2023-Soda,Mendez.Kampik-2025-EUMAS}, a statically typed descriptive language that provides executable implementations of these specifications\footnote{\texttt{\url{https://github.com/julianmendez/verification-pipeline}}}. The framework further build on the action language ${\cal C}_{MT}$\cite{brannstrom2026humanemotionverificationaction}, which in turn is built on earlier action languages \cite{gelfond1998action,dworschak2008mbn}.
    The proposed approach operationalizes action language semantics through executable verification components that process states, actions, transitions, and rules within a unified pipeline architecture. By organizing verification as compositional pipelines over trajectories and transition sequences, the framework supports clear specification, reusable reasoning components, and modular verification workflows. The framework further supports modular refinement of trajectory models through compositional rule extensions, allowing additional contextual information to be incorporated without changing the underlying verification structure. We illustrate this through examples involving misinformation propagation and emotional reasoning, where emotional fluents refine the interpretation and verification of agent behavior.

    The use of \Tiles to model verification pipelines provides several advantages. It offers a graphical overview of the specification, improving transparency and interpretability. Through its integration with \Soda, it enables execution within the Java Virtual Machine (JVM) ecosystem \cite{lindholm2013java}.
    The pipeline-based design ensures termination of the verification process by construction.
    In summary, the contributions of this paper are as follows:
    \begin{itemize}
        \item We introduce a modular framework for transition and trajectory verification based on \Tiles and \Soda.
        \item The framework operationalizes action language semantics through executable verification pipelines.
        \item We provide a YAML-based specification layer and parser for operationalizing domain descriptions and input sequences of states and actions (trajectories) within the verification framework.
        \item We show that our approach provides a transparent, executable, and scalable method for verifying rule-based system executions.
    \end{itemize}

    The remainder of this paper is structured as follows. In Section~\ref{sec:preliminaries}, we introduce the preliminaries, including \Soda, \Tiles, and the ${\cal C}_{MT}$ action language. In Section~\ref{sec:model}, we present the formal model underlying our framework. In Section~\ref{sec:operationalization}, we describe the operationalization of the model using executable specifications. In Section~\ref{sec:evaluation}, we outline the evaluation of the proposed approach. In Section~\ref{sec:background}, we discuss related work and position our contribution. See the implementation in the associated repository, and see the proofs in Appendix A.

    \section{Preliminaries}
    \label{sec:preliminaries}

    In this section, we introduce the preliminaries underlying our framework, including \Soda, \Tiles, and the ${\cal C}_{MT}$ action language.

    \Soda (Symbolic Objective Descriptive Analysis) is a statically typed descriptive language for specifying requirements using a small set of constructs. A definition in \Soda has the form
    $f(x_{1} : A_{1}) \ldots (x_{n} : A_{n}) : A = e,$
    where $f$ is a function name, $x_i$ are parameters of type $A_i$, and $e$ is an expression of type $A$. A function without parameters is a constant. Functions can be invoked using named parameters with the notation $f(x := v)$. The language includes typed lambda expressions of the form $\lambda (x : A) \rightarrow e$, conditional expressions $\texttt{if } b \texttt{ then } e_1 \texttt{ else } e_2$, and pattern matching constructs
    $\texttt{match } x \ \texttt{case } p_1 \Rightarrow e_1 \ \ldots \ \texttt{case } p_n \Rightarrow e_n,$
    where $p_i$ are constructor patterns. Standard arithmetic and logical operators are provided, and logical expressions are evaluated with lazy evaluation.

    \Soda uses immutable definitions without side effects and supports recursive functions over finite structures. It provides classes with constant and function definitions, as well as abstract declarations that must be instantiated. Classes can be extended via \texttt{extends} and support polymorphism through type parameters with optional subtype and supertype bounds.
    \Soda specifications can be translated into Scala for execution on the Java Virtual Machine, and a fragment can be translated into Lean to enable formal verification.

    \Tiles is a framework for composing functions into structured descriptions. A \emph{tile} is a function with input and output types, denoted as
    \tiles{\mathsf{_{\alpha} \ \mathsf{name} \ _{\beta} } },
    where $\alpha$ and $\beta$ are types (atomic, tuple, or sequence). Tiles may depend on contextual information and can be combined by composition: given $\tiles{_{\alpha} \ \mathsf{tile}_{1} \ _{\beta} }$ and $\tiles{_{\beta} \ \mathsf{tile}_{2} \ _{\gamma} }$, their composition yields $\tiles{_{\alpha} \ \mathsf{tile}_{1} \ \Big| \ _{\beta} \ \mathsf{tile}_{2} \ _{\gamma} }$, which is equivalent to $\tiles{_{\alpha} \ \mathsf{tile}_{2} \circ \mathsf{tile}_{1} \ _{\gamma} }$, where $\circ$ denotes composition of functions. Multiple tiles can be connected to form pipelines with a single entry point and a final output, typically a Boolean value.

    Primitive tiles (e.g., $\mathsf{bind}$, $\mathsf{fold}$) define basic transformations over sequences, while composite tiles are obtained by connecting tiles. Constants are tiles without input that provide deterministic values.
    Each configuration can be implemented in \Soda, where tiles correspond to functions and their compositions, which are in a statically-typed and immutable setting.

    \subsection{The ${\cal C}_{MT}$ Action Language}

    ${\cal C}_{MT}$ \cite{brannstrom2026humanemotionverificationaction} consists of a set of symbols representing actions and fluents, forming the alphabet of the action language.
    Let \sAction be a non-empty set of actions and \sFluent be a non-empty set of fluents.

    The ${\cal C}_{MT}$ domain description language $D^{MT}(\sAction, \sFluent)$
    consists of static and dynamic causal laws, defining how actions and state conditions constrain or influence state evolution in so-called trajectories of the form $\langle s_{0}, A_{1}, s_{1}, A_{2},$ $\dots,$ $A_{n}, s_{n} \rangle$ of $D^{MT}(\bf A, F)$, where $A_{i} \subseteq A$ is a set of actions and $s_{i}$ is a state of $D^{MT}(\bf A, F)$. These laws are of the following form:

    \begin{tabular}{ll}
        $(a \; \mathbf{causes} \; f_{1},\dots,f_{n} \; \mathbf{if} \; g_{1}, \dots, g_{m})$ & $(1)$ \\
        $(f_{1},\dots,f_{n} \; \mathbf{if} \; g_{1}, \dots, g_{m})$ & $(2)$  \\
        $(f_{1},\dots,f_{n} \; \mathbf{triggers} \; a)$ & $(3)$  \\
        $(f_{1},\dots,f_{n} \; \mathbf{allows} \; a)$ & $(4)$  \\
        $(f_{1},\dots,f_{n} \; \mathbf{inhibits} \; a)$ & $(5)$  \\
        $(\mathbf{no\text{-}concurrency} \; a_{1}, \dots, a_{n})$ & $(6)$  \\
        $(\mathbf{default} \; g)$ & $(7)$ \\[4pt]

        $(a \; \mathbf{influences} \; f_{1}, \dots, f_{n} \; \mathbf{if} \; g_{1}, \dots, g_{m})$ & $(8)$ \\
        $(g_{1}, \dots, g_{m} \; \mathbf{influences} \; f_{1},\dots,f_{n})$ & $(9)$  \\
        $(g_{1}, \dots, g_{m} \; \mathbf{facilitates} \; a)$ & $(10)$ \\
        $(g_{1}, \dots, g_{m} \; \mathbf{contravenes} \; a)$ & $(11)$ \\
        $(g_{1},\dots,g_{m}\; \mathbf{forbids\text{-}to\text{-}cause} \; f_{1}, \dots, f_{n})$ & $(12)$
    \end{tabular}

    \noindent
    where
    $a \in \sAction$,
    $a_{i} \in \sAction$,
    $f_{j},g_{j} \in \sFluent$.

    The rule forms of ${\cal C}_{MT}$ serve as the source language for the verification framework introduced in the current work. Instead of defining operational semantics directly over these laws, we translate them into a uniform transition-verification model based on rules, transitions, and trajectories. The corresponding verification semantics is therefore presented in Section~\ref{sec:model}.

    \section{Model}
    \label{sec:model}

    In this section, we present the formal model underlying the framework for representing and verifying trajectories of states and actions.

    The model is based on fluent values, actions, rules, transitions, and trajectories. Fluent values represent properties of the system, actions represent events that may change these properties, and rules specify causal and normative relations between states and actions, including triggering, inhibition, facilitation, influence, and concurrency constraints.

    A transition represents a single evolution step of the system and consists of an input state, a set of executed actions, and a resulting output state. A trajectory is a finite sequence of such transitions represented through alternating states and action sets over time. The purpose of the model is to verify whether the transitions induced by a trajectory satisfy the rules of the domain.

    Our fluent representation differs from a tuple $\langle$fluent name, value name$\rangle$~\cite{gelfond1998action}, by partitioning the fluent values and directly identifying the fluents with their state, and to ensure a uniform and unambiguous representation, fluent values are assumed to be globally unique.
    For example, in the setting of appraisal theory of emotion, fluents such as `goal' and `need' may be included, each with possible values such as `high', `undecided', and `low'. The corresponding fluent values may then be represented as `goal-high', `goal-undecided', `goal-low', `need-high', `need-undecided', and `need-low'. In this way, each fluent value uniquely determines the fluent and its value.

    \begin{definition}[Fluent Values and Fluents]
        \normalfont
        \label{def:fluent-values-and-fluents}
        Let $\mathsf{FV}$ be a non-empty finite set of identifiers, then the elements $v \in \mathsf{FV}$ are called \emph{fluent values}.

        We say that $\mathsf{F}$ is a \emph{set of fluents} for $\mathsf{FV}$ if $\mathsf{F}$ is a partition for $\mathsf{FV}$. In that case, each element of $\mathsf{F}$ is called a \emph{fluent}. In other words, if $F_{1} \in \mathsf{F}$ and $F_{2} \in \mathsf{F}$, then either $F_{1} = F_{2}$ or $F_{1} \cap F_{2} = \emptyset$. Conversely, we say that $\mathsf{FV}$ is the \emph{set of fluent values} of $\mathsf{F}$.
        We say that a set of fluent values $I \subseteq \mathsf{FV}$ is \emph{consistent} with respect to $\mathsf{F}$, denoted $I \subseteqFV$, if and only if every fluent value in $I$ corresponds to a different fluent in $\mathsf{F}$. When it is clear from the context that $I$ needs to be consistent, we omit $\mathsf{F}$ in the notation.
    \end{definition}

    $\mathsf{FV}$ is the domain of atomic symbols used to represent states, and a state will be a subset of $\mathsf{FV}$.
    The model includes actions and rule identifiers. Actions represent discrete events that may occur during a transition, while rule identifiers specify the type of constraint or relation to be evaluated over states and actions.

    \begin{definition}[Actions and Rule Identifiers]
        \normalfont
        \label{def:actions-and-role-identifiers}
        Let $\mathsf{A}$ be a non-empty finite set of identifiers, then the elements $a \in \mathsf{A}$ are called \emph{actions}. Let $\mathsf{RI}$ be the set
        $\mathsf{RI} = \{ \mathbf{causes\text{-}if},$ $\mathbf{if\text{-}rule},$ $\mathbf{triggers},$ $\mathbf{allows},$ $\mathbf{inhibits},$ $\mathbf{no\text{-}concurrency},$ $\mathbf{default},$ $\mathbf{influences\text{-}if},$ $\mathbf{influences},$ $\mathbf{facilitates},$ $\mathbf{contravenes},$ $\mathbf{forbids\text{-}to\text{-}cause} \}$,
        then the elements $s \in \mathsf{RI}$ are called \emph{rule identifiers}.
    \end{definition}

    $\mathsf{A}$ defines the set of actions that may occur in transitions and $\mathsf{RI}$ defines the set of rule types. A rule identifier determines how the components of a rule are interpreted in the verification function. No semantics is attached at this level, because the semantics is defined below by $\varphi_{\mathsf{verify}}$.

    Rules define constraints on states and actions using a uniform, tuple-based representation. Each rule specifies a set of input conditions, a set of relevant actions, and a set of resulting conditions, together with a type that determines how the rule is interpreted during verification.

    \begin{definition}[Rule, Transition, and Trajectory]
        \normalfont
        \label{def:rules}
        Let $\mathsf{F}$ be a non-empty finite set of fluents, $\mathsf{FV}$ the set of fluent values for $\mathsf{F}$, $\mathsf{A}$ a non-empty finite set of actions, and $\mathsf{RI}$ the set of rule identifiers as defined above. A \emph{rule} is a tuple $\langle s, I, A, O \rangle$ such that $s \in \mathsf{RI}$, $I \subseteqFV$ and $O \subseteqFV$, and $A \subseteq \mathsf{A}$.
        A \emph{transition} is a tuple $t = \langle t_{I}, t_{A}, t_{O} \rangle$ such that $t_{I} \subseteqFV$, $t_{O} \subseteqFV$, and $t_{A} \subseteq \mathsf{A}$ is a finite set of actions.
        The set of transitions over $\mathsf{FV}$ and $\mathsf{A}$ is $\mathsf{T} = \{ \langle t_{I}, t_{A}, t_{O}\rangle \mid t_{I} \subseteqFV, t_{A} \subseteq \mathsf{A}, t_{O} \subseteqFV \}$.
        A \emph{trajectory} is a finite sequence $\mathcal{SA} = \langle s_{0}, A_{1}, s_{1}, \ldots , A_{n}, s_{n} \rangle$ such that $s_{i} \subseteqFV$ for all $0 \leq i \leq n$ and $A_{i} \subseteq \mathsf{A}$ for all $1 \leq i \leq n$.
        Considering the fluents and actions, we can define its rules as follows:

        \begin{tabular}{ll}
            $\langle \mathbf{causes\text{-}if}, \{ g_{1}, \dots, g_{m} \}, \{ a \}, \{f_{1},\dots,f_{n} \} \rangle$ & $(1)$ \\
            $\langle \mathbf{if}, \{ g_{1}, \dots, g_{m} \}, \emptyset , \{f_{1},\dots,f_{n} \} \rangle$ & $(2)$  \\
            $\langle \mathbf{triggers}, \{f_{1},\dots,f_{n} \}, \{a\}, \emptyset \rangle$ & $(3)$  \\
            $\langle \mathbf{allows}, \{ f_{1},\dots,f_{n} \} , \{ a \}, \emptyset \rangle$ & $(4)$  \\
            $\langle \mathbf{inhibits} , \{f_{1},\dots,f_{n} \}, \{ a\} , \emptyset \rangle$ & $(5)$  \\
            $\langle \mathbf{no\text{-}concurrency} , \emptyset , \{a_{1}, \dots, a_{n}\}, \emptyset \rangle $ & $(6)$  \\
            $\langle \mathbf{default} , \{  g \} , \emptyset , \emptyset\} \rangle$ & $(7)$ \\
            $\langle \mathbf{influences\text{-}if}, \{ g_{1}, \dots, g_{m}\}, \{a \}, \{f_{1}, \dots, f_{n} \} \rangle$ & $(8)$ \\
            $\langle \mathbf{influences}, \{g_{1}, \dots, g_{m} \}, \emptyset , \{ f_{1},\dots,f_{n} \} \rangle$ & $(9)$  \\
            $\langle \mathbf{facilitates}, \{g_{1}, \dots, g_{m} \}, \{ a\}, \emptyset \rangle$ & $(10)$ \\
            $\langle \mathbf{contravenes}, \{g_{1}, \dots, g_{m} \}, \{ a\}, \emptyset \rangle$ & $(11)$ \\
            $\langle \mathbf{forbids\text{-}to\text{-}cause}, \{g_{1},\dots,g_{m}\}, \emptyset , \{f_{1}, \dots, f_{n}\} \rangle$ & $(12)$
        \end{tabular}

    \end{definition}

    The tuple structure is the same for all rule types. Its meaning is not encoded in the tuple itself but is determined entirely by the rule identifier $s$ through the verification function.
    This setup makes the input conditions, relevant actions, and resulting effects explicit, while leaving their interpretation to the rule identifier.
    A transition captures a single step with an input state $t_{I}$, the set of actions that occur $t_{A}$, and the resulting state $t_{O}$. All components are sets, and no assumption is made about how $t_{O}$ is produced from $t_{I}$, because that is left to the rule verification. We collect such steps into finite state-action sequences, which are the objects verified by the pipeline.
    The verification function determines whether a transition satisfies a given rule.

    \begin{definition}[Verification Function]
        \normalfont
        \label{def:verification-function}
        Let $\mathsf{FV}$ be a non-empty finite set of fluent values, $\mathsf{A}$ a non-empty finite set of actions, $\mathsf{RI}$ the set of rule identifiers, $\mathsf{R}$ a set of rules over $\mathsf{FV}, \mathsf{A}, \mathsf{RI}$, and $\mathsf{T}$ a set of transitions over $\mathsf{FV}, \mathsf{A}$. We define
        $\varphi_{\mathsf{inhibit}} : \mathsf{T} \times \mathsf{R} \to \mathcal{P}(\mathsf{A})$,
        $\varphi_{\mathsf{inhibit\text{-}set}} : \mathsf{T} \times \mathcal{P}(\mathsf{R}) \to \mathcal{P}(\mathsf{A})$, and
        $\varphi_{\mathsf{verify}} : \mathsf{T} \times \mathsf{R} \times \mathcal{P}(\mathsf{A}) \to \mathbb{B}$ as follows.
        For $t = \langle t_{I}, t_{A}, t_{O} \rangle \in \mathsf{T}$, $r = \langle s, I, A, O \rangle \in \mathsf{R}$,
        \[
            \varphi_{\mathsf{inhibit}} (t, r) =
            \begin{cases}
                A & \text{if } (I \subseteq t_{I}) \text{ and } s = \mathbf{inhibits} \text{ or } s = \mathbf{contravenes} \\
                \emptyset & \text{otherwise} \\
            \end{cases}
        \]

        Analogously,
        \[
            \varphi_{\mathsf{inhibit\text{-}set}} (t, \mathsf{R}) = \bigcup \limits _{r \in \mathsf{R}} \varphi _{\mathsf{inhibit}} (t, r)
        \]

        Let $B \subseteq \mathsf{A}$ such that $B = \varphi _{\mathsf{inhibit\text{-}set}} (\mathsf{R})$, then

        \[
            \varphi_{\mathsf{verify}} (t, r, B) =
            \begin{cases}
                (I \subseteq t_{I})
                \land (\lvert A \rvert = 1) \land (A \subseteq t_{A}) \land (A \cap B = \emptyset) \Longrightarrow (O \subseteq t_{O}) \\
                \qquad \text{ if } s = \mathbf{causes\text{-}if}  \text{ or } s = \mathbf{influences\text{-}if} \\
                (I \subseteq t_{I}) \Longrightarrow (O \subseteq t_{I})  \\
                \qquad \text{ if } s = \mathbf{if\text{-}rule} \text{ or } s = \mathbf{influences} \\
                (I \subseteq t_{I}) \land (\lvert A \rvert = 1) \land (A \cap B = \emptyset) \Longrightarrow (A \subseteq t_{A}) \\
                \qquad \text{ if } s = \mathbf{triggers} \\
                \text{true} \\
                \qquad \text{ if } s = \mathbf{allows} \text{ or } s = \mathbf{facilitates} \\
                (I \subseteq t_{I}) \Longrightarrow (\lvert A \rvert = 1) \land (A \cap t_{A} = \emptyset) \\
                \qquad \text{ if } s = \mathbf{inhibits} \text{ or } s = \mathbf{contravenes} \\
                \Longrightarrow \lvert A \cap t_{A} \rvert \le 1 \\
                \qquad \text{ if } s = \mathbf{no\text{-}concurrency} \\
                \Longrightarrow (\lvert I \rvert = 1) \land (I \subseteq t_{I}) \\
                \qquad \text{ if } s = \mathbf{default} \\
                (I \subseteq t_{I}) \Longrightarrow (O \cap t_{O} = \emptyset) \\
                \qquad \text{ if } s = \mathbf{forbids\text{-}to\text{-}cause}
            \end{cases}
        \]
        Each rule has a precondition and a postcondition connected by an arrow ($\Longrightarrow$). When the precondition is empty, we put the postcondition directly after the arrow (as in $\mathbf{no\text{-}concurrency}$ and $\mathbf{default}$).

    \end{definition}

    The function $\varphi_{\mathsf{verify}}(t,r,B)$ evaluates whether a transition $t$ satisfies a rule $r$, given a set of inhibited actions $B$. The function is defined case-wise on the rule identifier $s$. Each case specifies a logical condition that must hold.
    Conditions involve checking whether $I \subseteq t_{I}$ (rule is applicable), checking whether $A$ is executed or not, and enforcing constraints on $t_{O}$. The output is Boolean, and no partial satisfaction or scoring is used.

    \begin{definition}[Rule Consistency]
        \normalfont
        \label{def:rule-consistency}
        If the precondition in a rule $r$ is satisfied, we say that $r$ is \emph{active}, and otherwise $r$ is \emph{passive}.
        We say that a transition $t$ is \emph{valid} for a rule $r$ if and only if either $t$ does not make $r$ active, or $t$ makes the rule active and its postcondition is satisfied. By contrast, if $t$ makes $r$ active and the postcondition is not satisfied, then $t$ is \emph{invalid} for $r$.
        We say that two rules $r_{1}$ and $r_{2}$ are \emph{in conflict} if and only if for every transition $t$ that makes $r_{1}$ and $r_{2}$ active, $t$ is valid for either $r_{1}$ or $r_{2}$ and invalid for the other one.
        A set of rules $R$ is \emph{consistent} if and only if each pair $r_{1}, r_{2} \in R$ is not in conflict.
    \end{definition}

    Let us see an example of two rules in conflict. Consider a fluent $f$ with fluent values $f_{1}$ and $f_{2}$ and the rules $\mathbf{default} \ f_{1}$ and $\mathbf{default} \ f_{2}$. We see that every transition $t$ makes the $\mathbf{default}$ rules active as they have no preconditions. Since a fluent cannot have two different values in the same state, $r_{1}$ and $r_{2}$ are in conflict. As criteria for rule conflict resolution is not part of the framework, we only consider set of rules that are consistent. Since no overriding is allowed and the rules are consistent, the $\mathbf{allows}$ and $\mathbf{facilitates}$ rules are not necessary.

    \begin{proposition}[Soundness]
        \normalfont
        \label{prop:soundness}
        $\varphi_{\mathsf{verify}}$ maps the semantics of ${\cal C}_{MT}$.
    \end{proposition}

    \begin{proof}
        A trajectory in ${\cal C}_{MT}$ is valid if and only if every transition is valid.
        By construction, the function $\varphi_{\mathsf{verify}} (t, r, B)$ applies the rules defined in ${\cal C}_{MT}$ to every transition $t = \langle t_{I}, t_{A}, t_{O} \rangle \in \mathsf{T}$, rule $r$, and set of inhibited actions $B$. Since the set of rules is consistent, they can be applied independently. Thus, $\varphi_{\mathsf{verify}}$ verifies the validity of the whole trajectory.
        $\qed$
    \end{proof}

    The function $\varphi _{\mathsf{verify}}$ can be used to provide a detailed verification.

    \begin{definition}[Detailed Verification Function]
        \normalfont
        \label{def:detailed-verification-function}
        Let $\mathsf{T}$ be a set of transitions, $\mathsf{R}$ a set of rules, and $\mathsf{A}$ a set of actions. Let $\varphi_{\mathsf{verify\text{-}detailed}} : \mathsf{T} \times \mathsf{R} \times \mathcal{P}(\mathsf{A}) \to \mathsf{T} \times \mathsf{R} \times \mathcal{P}(\mathsf{A}) \times \mathbb{B}$ be defined as:
        $\varphi_{\mathsf{verify\text{-}detailed}} (t, r, B) = \langle t, r, B, \varphi_{\mathsf{verify}} (t, r, B) \rangle.$
    \end{definition}

    The function $\varphi_{\mathsf{verify\text{-}detailed}}$ lifts $\varphi_{\mathsf{verify}}$ by returning the transition $t$, the rule $r$, the inhibited actions $B$, and the Boolean result.
    This is a structural extension that preserves all inputs for downstream processing (e.g., mapping or aggregation), without changing the semantics of verification.

    Our framework allows defining properties between rules, including the rewriting of the $\mathbf{contravenes}$ rule, as shown below.

    \begin{proposition}[Properties of Contravenes]
        \normalfont
        \label{prop:contravenes-forbids-to-cause}
        Consider non-empty finite sets of fluent values $I, O$, an action $a$, consider the rules $r_{1} = \langle \mathbf{forbids\text{-}to\text{-}cause}, I, \emptyset, O \rangle$, $r_{2} = \langle \mathbf{causes\text{-}if}, I, \{a\}, O \rangle$, and $r_{3} = \langle \mathbf{contravenes}, I, \{a\}, \emptyset \rangle$. Consider that $a$ is not inhibited for $r_{1}$ and $r_{2}$. If a transition $t$ is valid for $r_{1}$ and for $r_{2}$, then it is valid for $r_{3}$, and if $t$ is invalid for $r_{1}$ and for $r_{2}$, then it is invalid for $r_{3}$.
    \end{proposition}

    \begin{proof}
        See the proof in Appendix A.
    \end{proof}

    Intuitively, Proposition~\ref{prop:contravenes-forbids-to-cause} shows that a $\mathbf{contravenes}$ rule can be replaced with the rules $\mathbf{forbids\text{-}to\text{-}cause}$ and $\mathbf{causes\text{-}if}$. That can be done by having a new fluent $g$ with a fluent value $g_{0}$ which must not occur in the rules nor in the transitions. After that, the transitions need to be updated with the following function $f_{I,O} : \mathsf{T} \to \mathsf{T}$,
    \[
        f_{I,O}(\langle t_{I}, t_{A}, t_{O} \rangle) =
        \begin{cases}
            \langle t_{I}, t_{A}, t_{O} \cup \{g_{0}\} \rangle & \text{if } I \subseteq t_{I} \\
            \langle t_{I}, t_{A}, t_{O} \rangle & \text{otherwise} \\
        \end{cases}
    \]

    The core semantic components are collected in the following transition verification model.

    \begin{definition}[Transition Verification Model]
        \normalfont
        \label{def:transition-verification-model}
        A \emph{transition verification model} is a tuple $\mathcal{M} = \langle \mathsf{F}, \mathsf{A}, \mathsf{RI}, \mathsf{R}, \varphi_{\mathsf{verify}} \rangle,$ where $\mathsf{F}$ is a non-empty finite set of fluents, $\mathsf{FV}$ is the set of fluent values for $\mathsf{F}$, $\mathsf{A}$ is a non-empty finite set of actions, $\mathsf{RI}$ is a non-empty finite set of rule identifiers as in Definition~\ref{def:actions-and-role-identifiers}, $\mathsf{R}$ is a non-empty finite set of rules, and $\varphi_{\mathsf{verify}}$ the verification function for $\mathsf{RI}$ as in Definition~\ref{def:verification-function}.
    \end{definition}

    In order to extract transitions from a trajectory, we introduce an auxiliary structure that incrementally collects elements and assembles them into triples.

    \begin{definition}[Sliding Window for Transitions]
        \normalfont
        \label{def:sliding-window}
        Consider the following trajectory $\mathcal{SA} = \langle s_{0},A_{1},s_{1},\ldots,A_{n},s_{n}\rangle$ such that $s_{i} \subseteq \mathsf{FV}$ for all $0 \leq i \leq n$ and $A_i \subseteq \mathsf{A}$ for all $1 \leq i \leq n$. A \emph{sliding window} is a tuple $w = \langle s,a,S\rangle$ where $s \subseteq \mathcal{P}(\mathsf{FV})$, $a \subseteq \mathcal{P}(\mathsf{A})$, $\lvert s\rvert \leq 1$, $\lvert a\rvert \leq 1$, and $S$ is a finite sequence of transitions. We denote $\mathsf{W}$ as the set of all such sliding windows.
    \end{definition}

    The sliding window stores partial information about a transition, consisting of at most one state component, at most one action component, and a sequence of completed transitions.

    We define a function that updates the sliding window by consuming elements of a trajectory and constructing transitions when sufficient information is available.

    \begin{definition}[Transition Construction Function]
        \normalfont
        \label{def:transition-construction-function}
        Let $\mathsf{W}$ be the set of sliding windows, and let $\mathsf{A}$ and $\mathsf{FV}$ be the sets of actions and fluent values, respectively. Let
        $\varphi_{\mathsf{next}} : \mathsf{W} \times \mathcal{P}(\mathsf{A}\cup \mathsf{FV}) \to \mathsf{W}$
        be a function defined as follows:
        \[
            \varphi_{\mathsf{next}}(\langle s, a, S\rangle, q) =
            \begin{cases}
                \langle q,\, a,\, S \rangle, & \text{ if } s = \emptyset, \\

                \langle s,\, q,\, S\rangle, & \text{ if } s \neq \emptyset \text{ and } a = \emptyset, \\

                \langle q,\, \emptyset,\, S \mathbin{\|} \langle s,a,q \rangle \rangle,
                & \text{ if } s \neq \emptyset \text{ and } a \neq \emptyset .
            \end{cases}
        \]
    \end{definition}

    The function $\varphi_{\mathsf{next}}$ processes one element at a time, either extending the current partial information or completing a transition and appending it to the sequence.
    To initialize the construction, we define an initial sliding window.

    \begin{definition}[Initial Sliding Window]
        \normalfont
        \label{def:initial-sliding-window}
        Let $\mathsf{W}$ be the set of sliding windows, then the initial sliding window $z_{\mathsf{init}} \in \mathsf{W}$ is defined as
        %\[
        $z_{\mathsf{init}} = \langle \emptyset, \emptyset, [~] \rangle.$
        %\]
    \end{definition}

    The initial sliding window provides the starting configuration for constructing transitions from a trajectory. Starting from $z_{\mathsf{init}}$, the function $\varphi_{\mathsf{next}}$ is applied sequentially to each trajectory element, incrementally producing a sequence of transitions.

    \begin{definition}[Sliding Window Sequence Construction]
        \normalfont
        \label{def:sliding-window-sequence-construction}
        Let $\mathcal{SA}$ $=$ $\langle t_{1},$ $\dots,$ $t_{m} \rangle$ be a sequence such that $t_k \subseteq \mathsf{FV}$ or $t_{k} \subseteq \mathsf{A}$ for all $1 \leq k \leq m$. Let $z_{\mathsf{init}} \in \mathsf{W}$ be the initial sliding window, and let $\varphi_{\mathsf{next}}$ be the transition construction function. The sequence $(w_{k})_{k=0}^{m}$ is defined by $w_{0} = z_{\mathsf{init}}$ and $w_{k} = \varphi_{\mathsf{next}}(w_{k-1}, t_{k})$ for all $1 \leq k \leq m$.
    \end{definition}

    This defines an inductive construction where each element of the trajectory is incorporated into the sliding window, producing a sequence of intermediate windows.
    Finally, we extract the sequence of constructed transitions from the sliding window, yielding the sequence of transitions corresponding to the original trajectory, which is subsequently used for rule verification.

    \begin{definition}[Transition Sequence Extraction]
        \normalfont
        \label{def:transition-sequence-extraction}
        Let $\varphi_{\mathsf{get\text{-}seq}}$ be the projection that maps a sliding window
        $\langle s, a, S \rangle$ to the sequence of transitions $S$.
    \end{definition}

    The above definitions specify a transformation from a trajectory to a sequence of transitions. This transformation is defined independently of any implementation and serves as the formal basis for the verification procedure.

    \begin{figure*}[ht!]
        \centering
        \begin{tikzpicture}[x=1mm, y=1mm, box/.style={rectangle, draw, rounded corners=2mm, minimum width=4mm, minimum height=8mm, align=center}]
            \node[box] (S) {};

            \node[box, right=6mm of S] (trajectory) {$\makecell[l]{\ \mathsf{trajectory} \, \\ \ } \Big| \ \makecell[l]{\mathsf{\qquad fold \ \mathnormal{z}_{init} \ using \ \varphi_{\mathsf{next}} } \  \\ \mathsf{ \ _{(t)} } } \Big| \makecell{\ \ \mathsf{apply} \ \varphi _{\mathsf{get\text{-}seq}} \\ \mathsf{\ _{\langle (s_{1}), (a_{1}), (\langle s_{2}, a_{2}, s_{3} \rangle ) \rangle} \qquad \ _{(\langle s_{2}, a_{2}, s_{3} \rangle )} } }  $};

            \node[box, below=8mm of S, xshift=18mm] (rules) {$\makecell{\ \mathsf{rules}\\ \qquad \qquad  \mathsf{ \ _{(r_{1})}}}$};

            \node[box, below=24mm of S, xshift=62mm] (preprocessor) {$\mathsf{\makecell{\qquad \mathsf{map \ \varphi _{inhibit\text{-}set}} \\ \mathsf{_{(\langle s_{2}, a_{2}, s_{3} \rangle ) ; (r_{1}) } \qquad } } \ \Big| \ \makecell{\qquad \mathsf{cross} \\ \mathsf{_{(\langle \langle s_{2}, a_{2}, s_{3} \rangle, (a_{3}) \rangle ), (r_{1})} \qquad } } \Big| \ \makecell{\mathsf{map} \ \varphi _{\mathsf{verify\text{-}detailed}} \\ \mathsf{_{(\langle \langle \langle s_{4}, a_{5}, s_{5} \rangle , (a_{6}) \rangle, r_{2} \rangle )} \qquad _{(\langle \langle s_{4}, a_{5}, s_{5} \rangle , (a_{6}) , r_{2} , b_{1} \rangle )}} }}$};

            \node[box, below=40mm of S, xshift=62mm] (postpreprocessor) {$\makecell{\mathsf{map} \ \varphi _{\mathsf{get\text{-}b}} \\ \mathsf{\ _{(\langle \langle s_{4}, a_{5}, s_{5} \rangle , (a_{6}) , r_{2} , b_{1} \rangle )} \qquad } } \ \Big| \ \makecell{\mathsf{fold \ \mathnormal{true} \ using \ \varphi _{\mathsf{all\text{-}true}}} \\ \mathsf{ \ _{(b_{1})} \qquad \qquad \qquad \qquad \qquad \ _{b_{2}}} }$};

            \draw [->] (S.east) -- (trajectory.west);
            \draw [->] (S.east) -- ++(2mm, 0mm) |- (rules.west);
            \draw [->] (rules.east) -- ++(3mm, 0mm) -- ++(0mm, -7mm) -- ++(-34mm, 0mm) |- (preprocessor.west);
            \draw [->] (trajectory.east) -- ++(3mm, 0mm) -- ++(0mm, -26mm) -- ++(-110mm, 0mm) |- (preprocessor.west);
            \draw [->] (preprocessor.east) -- ++(3mm, 0mm) -- ++(0mm, -8mm) -- ++(-105mm, 0mm) |- (postpreprocessor.west);

            \draw [thick] (10,-45) circle (1.3);
            \draw [thick] (11,-46) -- (12,-47);
            \draw [->, dotted] (18,-45) -- (13, -45);

        \end{tikzpicture}
        \caption{Pipeline to check a trajectory.}
        \label{fig:pipeline-trajectory}
    \end{figure*}
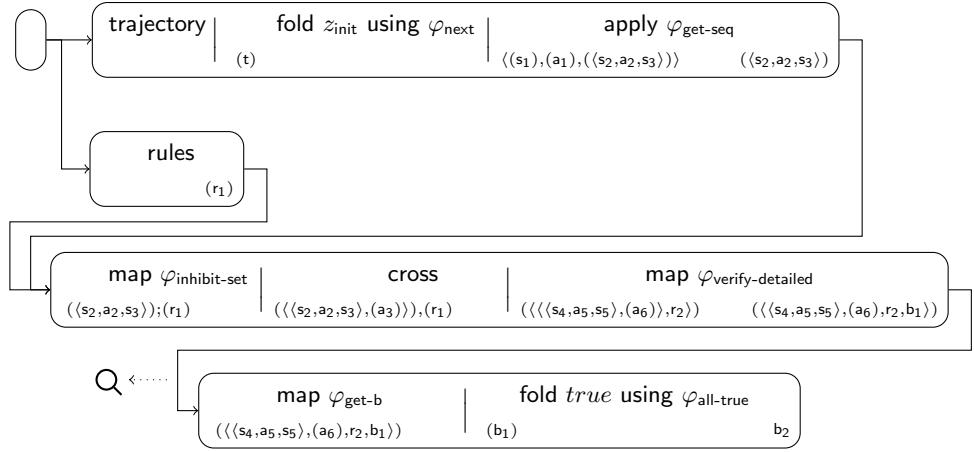

    In Figure~\ref{fig:pipeline-trajectory}, we present the compositional pipeline induced by the transformation defined above, where a trajectory is processed sequentially to construct a sequence of transitions, which are subsequently combined with rules to perform verification. The pipeline is structured as a sequence of typed transformations over trajectories, transitions, and rules, where each tile corresponds to a formally defined operation. It operates over Boolean values $\mathsf{b}$, sets of fluent values $\mathsf{s}$, sets of actions $\mathsf{a}$, and rules $\mathsf{r}$; elements $\mathsf{t}$ denote either sets of fluent values or sets of actions, while $(\cdot)$ and $\langle \cdot , \ldots , \cdot \rangle$ denote sequences and tuples, respectively.
    The core tiles in the pipeline correspond to the following operations:

    \begin{tabular}{p{0.48\linewidth} p{0.48\linewidth}}
        $\tiles{\makecell{\mathsf{trajectory} \\ \qquad \qquad \qquad \mathsf{_{(t)} }}}$
        &
        is a constant tile that provides all the elements in a trajectory;
        \\[6pt]

        $\tiles{\makecell{\mathsf{fold \ \mathnormal{z}_{init} \ using \ \varphi_{\mathsf{next}}  } \\ \mathsf{\ _{(t)} \qquad \qquad _{\langle (s_{1}) , (a_{1}), (\langle s_{2}, a_{2}, s_{3} \rangle ) \rangle } }} }$
        &
        is a $\mathsf{fold}$ tile that takes the elements $t$ in a trajectory and builds a sequence of transitions inside a sliding window,
        \\[6pt]

        $\tiles{\makecell{\mathsf{apply} \ \varphi_{\mathsf{get\text{-}seq}} \\ \mathsf{_{\langle (s_{1}), (a_{1}), (\langle s_{2}, a_{2}, s_{3} \rangle ) \rangle} \ \qquad \ _{(\langle s_{2}, a_{2}, s_{3} \rangle )} } } }$
        &
        is an $\mathsf{apply}$ tile that retrieves the sequence of transitions;
        \\[6pt]

        $\tiles{\makecell{\mathsf{rules} \\ \qquad \qquad \mathsf{\ _{(r)}}}}$
        &
        is a constant tile that provides all the rules to be applied;
        \\[6pt]

        $\tiles{\makecell{\qquad \mathsf{map \ \varphi _{inhibit\text{-}set}} \\ \mathsf{_{(\langle \langle s_{1}, a_{1}, s_{2} \rangle); (r)} \qquad _{(\langle \langle s_{1}, a_{1}, s_{2} \rangle, (a_{2})\rangle); (r)} } } }$
        &
        is a $\mathsf{map}$ tile that computes the prepreprocessing of the transitions, collecting inhibited actions, i.e. the sequence $\mathsf{(a_{2})}$ contains all the actions from $\mathbf{inhibits}$ or $\mathbf{contravenes}$ rules, the sequence $\mathsf{(r)}$ is used by $\varphi _{\mathsf{inhibit\text{-}set}}$, but remains as a constant for the tile;
        \\[6pt]

        $\tiles{\makecell{\mathsf{cross} \\ \mathsf{_{(\alpha_{1}) , (\beta_{1})} \qquad  _{(\langle \alpha_{2}, \beta_{2} \rangle)}}}}$
        &
        is a tile that computes the cross product between the elements of the sequences $(\alpha_{1})$ and $(\beta_{1})$;
        \\[6pt]

        $\tiles{\makecell{\mathsf{map} \ \varphi _{\mathsf{verify\text{-}detailed}} \\ \mathsf{_{(\langle \langle \langle s_{1}, a_{1}, s_{2} \rangle ,  (a_{2}) \rangle , r \rangle )} \qquad \! \! _{(\langle \langle s_{1}, a_{1}, s_{2} \rangle , (a_{2}) , r , b \rangle )}}}}$
        &
        is a $\mathsf{map}$ tile that verifies each pair transition-rule considering the inhibited actions. %, where the 
        Input is preserved in the output to allow for a detailed report;
        \\
    \end{tabular}

    The pipeline outputs an array of Boolean as output indicating which transitions are violated. If only a yes/no answer is needed, it is possible to add a tile that aggregates the whole sequence to verify that all values are true.

    \begin{tabular}{p{0.48\linewidth} p{0.48\linewidth}}
        $\tiles{\makecell {\mathsf{map \ \varphi _{\mathsf{get\text{-}b}}}  \\ \mathsf{ \ _{(\langle \langle s_{1}, a_{1}, s_{2} \rangle , (a_{2}) , r , b \rangle )} \ \qquad \ _{(b)}}}}$
        &
        is a $\mathsf{map}$ tile that applies $\varphi _{\mathsf{get\text{-}b}}$, where $\varphi _{\mathsf{get\text{-}b}}$ is a function that retrieves the last element of the tuple;
        \\[6pt]

        $\tiles{\makecell{\mathsf{fold \ \mathnormal{true} \ using \ \varphi _{\mathsf{all\text{-}true}}} \\ \mathsf{ \ _{(b_{1})} \qquad \qquad \qquad \qquad \qquad \ _{b_{2}}}}}$
        &
        is a $\mathsf{fold}$ tile that checks whether all the elements in the sequence are true, where $\varphi _{\mathsf{all\text{-}true}} : \sBoolean \times \sBoolean \to \sBoolean$ is defined as:
        %\[
        $\varphi _{\mathsf{all\text{-}true}} (b_{1} , b_{2}) = b_{1} \land b_{2} \, .$
        %\]
        \\
    \end{tabular}

    The pipeline thus provides an executable realization of the previously defined verification functions. Its compositional structure ensures that the verification procedure remains semantically grounded, while supporting modular reuse and extensibility of individual components without compromising correctness.

    We denote by $\Pi_{\mathsf{traj}}$ the tile composition displayed in Figure~\ref{fig:pipeline-trajectory}.

    \begin{proposition}[Termination]
        \normalfont
        \label{prop:pipeline-termination}
        Let $\mathcal{SA} = \langle s_{0},A_{1},s_{1},\ldots,A_{n},s_{n}\rangle$ be a finite trajectory and let $\mathsf{R}_{D}\subseteq\mathsf{R}$ be a finite set of rules. Then the following holds:

        \begin{enumerate}
            \item $\Pi_{\mathsf{traj}}(\mathcal{SA},\mathsf{R}_{D})$ is well-defined and returns a finite sequence of Boolean values.

            \item The composition obtained by appending the tile $\tiles{\mathsf{fold}\ \mathsf{true}\ \mathsf{using}\ \varphi_{\mathsf{all\text{-}true}}}$ to $\Pi_{\mathsf{traj}}$ returns a Boolean value.
        \end{enumerate}
    \end{proposition}

    \begin{proof}
        Since $\mathcal{SA}$ is finite, $\tiles{\mathsf{fold}\ z_{\mathsf{init}}\ \mathsf{using}\ \varphi_{\mathsf{next}}}$ performs finitely many applications of $\varphi_{\mathsf{next}}$ and therefore terminates. Hence, $\tiles{\mathsf{apply}\ \varphi_{\mathsf{get\text{-}seq}}}$ returns a finite transition sequence.
        Since $\mathsf{R}_{D}$ is finite, $\tiles{\mathsf{cross}}$ returns a finite sequence of transition-rule pairs. The tiles $\tiles{\mathsf{map}\ \varphi_{\mathsf{inhibit}}}$, $\tiles{\mathsf{map}\ \varphi_{\mathsf{verify\text{-}detailed}}}$, and $\tiles{\mathsf{map}\ \varphi_{\mathsf{get\text{-}b}}}$ apply total functions to finite sequences and therefore terminate. Consequently, $\Pi_{\mathsf{traj}}(\mathcal{SA},\mathsf{R}_{D})$ is well-defined and returns a finite sequence of Boolean values.
        For the aggregated case, $\tiles{\mathsf{fold}\ \mathsf{true}\ \mathsf{using}\ \varphi_{\mathsf{all\text{-}true}}}$ performs finitely many applications of $\varphi_{\mathsf{all\text{-}true}}$ over this finite Boolean sequence. Therefore, the resulting composition returns a Boolean value.
        $\qed$
    \end{proof}

    One of the advantages of the \Tiles framework is the possibility of computing complexity by looking at the pipeline.

    \begin{proposition}[Complexity]
        \normalfont
        \label{prop:complexity}
        If $r$ is the number of rules, $t$ is the number of states and sets of actions in a trajectory, and $d$ is a bound for the number of steps to evaluate $\varphi _{\mathsf{verify\text{-}detailed}}$, then the complexity is $\mathcal{O}(r \times t \times d)$.
    \end{proposition}

    \begin{proof}
        Visiting each tile in Figure~\ref{fig:pipeline-trajectory}, we see that $\mathsf{trajectory}$ is $\mathcal{O}(t)$, which is kept through $\mathsf{fold}$ and $\mathsf{apply}$, while $\mathsf{rules}$ is $\mathcal{O}(r)$. When the elements of $\mathsf{rules}$ and $\mathsf{trajectory}$ combine in $\mathsf{cross}$, the complexity is $\mathcal{O}(r \times t)$. Since, $d$ is a bound for $\varphi _{\mathsf{verify\text{-}detailed}}$, the complexity of $\mathsf{map} \ \varphi _{\mathsf{verify\text{-}detailed}}$ is $\mathcal{O}(r \times t \times d)$, which is kept up to the end of the pipeline.
        $\qed$
    \end{proof}

    We show the operationalization of this framework, especially how domain descriptions and trajectories can be specified in a concrete format and executed within the system.

    \section{Operationalization}
    \label{sec:operationalization}

    To provide an executable representation of the formal model introduced in Section~\ref{sec:model}, we operationalize its structures through YAML specifications. A YAML specification encodes a transition verification model $\mathcal{M} = \langle \mathsf{F}, \mathsf{A}, \mathsf{RI}, \mathsf{R}, \varphi_{\mathsf{verify}} \rangle$ together with a trajectory $\mathcal{SA} = \langle s_{0}, A_{1}, s_{1}, \ldots, A_{n}, s_{n} \rangle.$

    The top-level YAML entries \texttt{fluents}, \texttt{actions}, \texttt{rules}, and \texttt{trajectory} correspond to the formal components of the model by, receptively, encoding the set $\mathsf{F}$ of fluents; the set $\mathsf{A}$ of actions; the set $\mathsf{R}$ of rules; and a trajectory $\mathcal{SA}$. Also, $\mathsf{FV}$ is the set of fluent values of $\mathsf{F}$.

    We use so-called \emph{validators} to check instances before entering the \Tiles pipeline. Fluent validators ensure valid and non-conflicting fluent values, action validators ensure valid actions, rule validators ensure rule consistency, trajectory validators ensure well-formed alternating state-action sequences, and configuration validators ensure consistency of the complete specification.

    The parameters used in rule specifications are typed as follows:

    \begin{itemize}
        \item $I$ : (\texttt{input}) set of fluent values ($I \subseteqFV$);

        \item $a$ : (\texttt{action}) single action ($a \in \mathsf{A}$);

        \item $O$ : (\texttt{output}) set of fluent values ($O \subseteqFV$);

        \item $B$ : (\texttt{inhibited}) set of inhibited or contravened actions ($B \subseteq \mathsf{A}$);

        \item $t_{I}$ : (\texttt{t.input}) input state of a transition ($t_{I} \subseteqFV$);

        \item $t_{A}$ : (\texttt{t.action}) action component of a transition ($t_{A} \subseteq \mathsf{A}$);

        \item $t_{O}$ : (\texttt{t.output}) output state of a transition ($t_{O} \subseteqFV$);

        \item $f$ : (\texttt{input\_fluent}) single fluent value ($f \in \mathsf{FV}$);

        \item $A$ : (\texttt{actions}) set of actions ($A \subseteq \mathsf{A}$).
    \end{itemize}

    Each rule entry in \texttt{rules} encodes a tuple $ \langle s,I,A,O\rangle \in \mathsf{R} $ according to its constructor type:

    \begin{itemize}
        \item \texttt{CausesIfRule} (\texttt{input}, \texttt{action}, \texttt{output}) encodes $ \langle \mathbf{causes\text{-}if}, I, \{a\}, O \rangle ; $

        \item \texttt{IfRule} (\texttt{input}, \texttt{output}) encodes $ \langle \mathbf{if\text{-}rule}, I, \emptyset, O \rangle ; $

        \item \texttt{TriggersRule} (\texttt{input}, \texttt{action}) encodes $ \langle \mathbf{triggers}, I, \{a\}, \emptyset \rangle ; $

        \item \texttt{AllowsRule} (\texttt{input}, \texttt{action}) encodes $ \langle \mathbf{allows}, I, \{a\}, \emptyset \rangle ; $

        \item \texttt{InhibitsRule} (\texttt{input}, \texttt{action}) encodes $ \langle \mathbf{inhibits}, I, \{a\}, \emptyset \rangle ; $

        \item \texttt{NoConcurrencyRule} (\texttt{actions}) encodes $ \langle \mathbf{no\text{-}concurrency}, \emptyset, A, \emptyset \rangle ; $

        \item \texttt{DefaultRule} (\texttt{input\_fluent}) encodes $ \langle \mathbf{default}, \{f\}, \emptyset, \emptyset \rangle ; $

        \item \texttt{InfluencesIfRule} (\texttt{input}, \texttt{action}, \texttt{output}) encodes $ \langle \mathbf{influences\text{-}if}, I, \{a\}, O \rangle ; $

        \item \texttt{InfluencesRule} (\texttt{input}, \texttt{output}) encodes $ \langle \mathbf{influences}, I, \emptyset, O \rangle ; $

        \item \texttt{FacilitatesRule} (\texttt{input}, \texttt{action}) encodes $ \langle \mathbf{facilitates}, I, \{a\}, \emptyset \rangle ; $

        \item \texttt{ContravenesRule} (\texttt{input}, \texttt{action}) encodes $ \langle \mathbf{contravenes}, I, \{a\}, \emptyset \rangle ; $

        \item \texttt{ForbidsToCauseRule} (\texttt{input}, \texttt{output}) encodes $ \langle \mathbf{forbids\text{-}to\text{-}cause}, I, \emptyset, O \rangle .  $
    \end{itemize}

    The parameters \texttt{input} and \texttt{output} encode sets of fluent values, \texttt{actions} encodes a set of actions, \texttt{action} encodes a single action, and \texttt{input\_fluent} encodes a single fluent value.
    The entry \texttt{trajectory} encodes an alternating sequence $\langle s_{0}, A_{1}, s_{1}, \ldots, A_{n}, s_{n} \rangle$ where each \texttt{state} entry encodes a subset $s_i \subseteq \mathsf{FV}$ of fluent values, and each \texttt{actions} entry encodes a subset $A_i \subseteq \mathsf{A}$ of actions. The first and last entries are \texttt{state} entries, and consecutive entries alternate between \texttt{state} and \texttt{actions}, i.e., no two consecutive \texttt{state} entries and no two consecutive \texttt{actions} entries may occur.

    \section{Evaluation}
    \label{sec:evaluation}

    The proposed model can be applied in several settings depending on whether rules, actions, and fluent values are known or unknown. It supports reasoning about partially observable states, hidden or inferred actions, and the validation of unknown governing rules from observed trajectories. Here, we focus on the setting where the rules and actions are \emph{known}, but fluent values are partially \emph{unknown}. To illustrate this, we present a misinformation-spread example in which a trajectory initially appears invalid when emotional states are omitted, but becomes valid once emotional fluents are considered, demonstrating the non-monotonic impact of extending the model with richer representations.

    \begin{example}[Modeling Misinformation on Social Networks]
        In our model, multiple agents share news items as messages within a network. Each agent may read truthful and false information. Agents share messages with all their contacts. Every agent that reads the truth can perfectly distinguish between true and false information.
        The only fluent is \emph{information}, and the values of the \emph{information} fluent are \emph{informed} and \emph{uninformed}. The available actions are \emph{share-truth}, \emph{share-lie}, \emph{read-truth}, and \emph{read-lie}.
        The modeled rules are:

        \begin{enumerate}
            \item Being \emph{informed} \textbf{contravenes} \emph{share-lie}.
            \item There is \textbf{no concurrency} between \emph{read-truth} and \emph{read-lie}.
            \item There is \textbf{no concurrency} between \emph{share-truth} and \emph{share-lie}.
        \end{enumerate}

        \noindent The following trajectory is not valid: $\langle \{\emph{informed}\}, \{\emph{share-lie}\}, \{\emph{informed}\} \rangle $
    \end{example}

    While the previous example only considers informational states, we now extend the model with emotional fluents to capture how emotional alignment affects the validity of trajectories.

    \begin{example}[Modeling Emotional States]
        \label{ex:emotional-states}
        Inspired by recent psychological research on misinformation diffusion and cognitive dissonance~\cite{wang2020fake,ecker2022psychological}, we use emotional states to model how inconsistencies between beliefs and actions may produce emotional conflict. For that, we have the fluent of emotional \emph{alignment}. The fluent can have two possible values: emotionally \emph{conflicted} and emotionally \emph{aligned}. This alignment can influence in the actions that an agents perform.

        \begin{enumerate}
            \item Being \emph{informed} and \emph{aligned} \textbf{contravenes} \emph{share-lie} (this rule is extended).

            \item There is \textbf{no concurrency} between \emph{read-truth} and \emph{read-lie}.

            \item There is \textbf{no concurrency} between \emph{share-truth} and \emph{share-lie}.

            \item Being \emph{informed} and performing \emph{share-lie} \textbf{influences} becoming \emph{conflicted}.

            \item Being \emph{informed} and performing \emph{share-truth} \textbf{influences} becoming \emph{aligned}.

            \item Being \emph{informed} and \emph{conflicted} \textbf{facilitates} \emph{share-lie}.

            \item Being \emph{informed} and performing \emph{read-lie} \textbf{influences} becoming \emph{conflicted}.

            \item Being \emph{informed} and performing \emph{read-truth} \textbf{influences} becoming \emph{aligned}.
        \end{enumerate}

        \noindent With these extended rules, and adding the fluent \emph{conflicted}, the trajectory becomes valid: $\langle \{\emph{informed}, \emph{conflicted}\}, \{\emph{share-lie}\}, \{\emph{informed}, \emph{conflicted}\} \rangle$
    \end{example}

    These examples illustrate how the framework supports modular refinement of trajectory models through compositional rule extensions. Additional contextual information, such as emotional states, can be incorporated without changing the underlying verification pipeline, enabling more expressive interpretations of agent behavior while preserving the same verification structure.

    The pipeline is optimized to use lazy evaluation on the set of rules and efficiently simplify the evaluation conditions in the rules. An instance of Example~\ref{ex:emotional-states} executed on a standard multi-core Linux computer (Intel Core Ultra 5 135H, $1600$ MHz $\times $ $6$ + $400$ MHz $\times $ $12$) shows the values in Figure~\ref{fig:comparison}. The columns show the time difference between an instance of 1 rule and an instance with the same trajectory but with 8 rules. The trajectory was generated by cyclically replicating a pattern.
    The measurements are linear following the theoretical expectation of Proposition~\ref{prop:complexity}, with some difference, which could be explained by memory management overhead.
    %due to optimizations at different levels of execution.

    \begin{figure*}
        \centering
        \includegraphics[width=\textwidth]{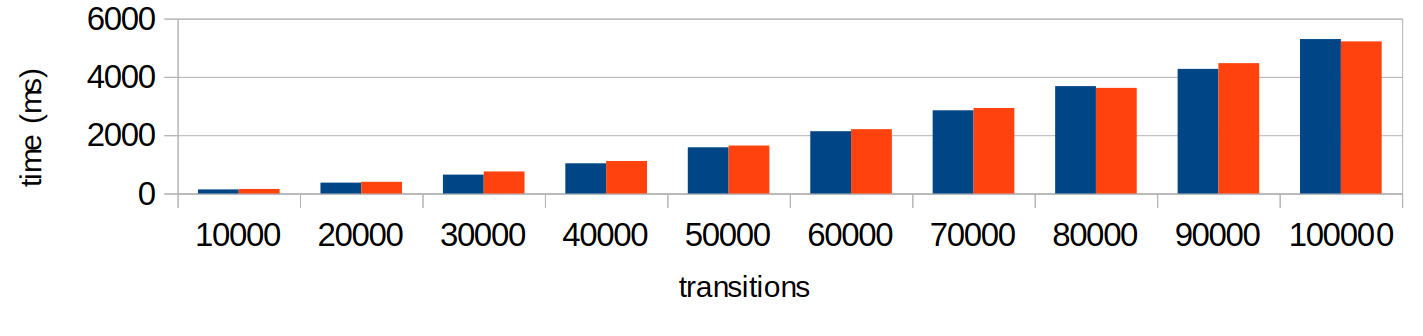}
        \caption{The columns compare execution times in milliseconds for the same number of transitions in a trajectory with different number of rules. The left side (blue) of each column is for 1 rule and the right side (red) of each column is for 8 rules.}
        \label{fig:comparison}
    \end{figure*}

    \section{Conclusion and Future Work}
    \label{sec:background}

    Building on seminal and recent research on reasoning about action and change~\cite{dworschak2008mbn,gelfond1998action,brannstrom2026humanemotionverificationaction}, existing implementations are commonly based on logic-programming systems such as GOLOG~\cite{Levesque-1997-Golog}, STRIPS~\cite{Fikes-1971-Strips}, and DLV$^{\text{K}}$~\cite{Eiter-2003-DVLK}. In this work, we introduced a modular and executable framework for transition and trajectory verification based on functional programming, implemented as verification pipelines in \Tiles~\cite{Mendez.Kampik.Aler.Dignum-2024-SCAI,Mendez.Kampik-2025-LNGAI} and executable specifications in \Soda~\cite{Mendez-2023-Soda,Mendez.Kampik-2025-EUMAS}.
    This approach supports reusable reasoning components, transparent execution workflows, concurrent rule evaluation, guaranteed pipeline termination, and YAML-based executable specifications. In this manner, the framework contributes toward bridging action reasoning and planning research towards compositional software engineering approaches~\cite{harmelen2021modular,mishra2024scalable,rudra2025composable} in functional and object-oriented programming, enabling integration with mainstream JVM-based software ecosystems~\cite{lindholm2013java}.

    Although this work has been limited to verification over explicitly specified trajectories and domain descriptions, and does not address automated trajectory generation or planning,
    it provides a full semantic embedding of the underlying action language ${\cal C}_{MT}$, and is extensible by design.

    \begin{credits}
        \subsubsection{\ackname} This work was partially supported by the Wallenberg AI, Autonomous Systems and Software Program (WASP) funded by the Knut and Alice Wallenberg Foundation.

    \end{credits}
%
% ---- Bibliography ----
%
% BibTeX users should specify bibliography style 'splncs04'.
% References will then be sorted and formatted in the correct style.
%
    \bibliographystyle{splncs04}
    \bibliography{main}

@inproceedings{Mendez.Kampik.Aler.Dignum-2024-SCAI,
    author = {Mendez, Julian Alfredo and Kampik, Timotheus and Aler Tubella, Andrea and Dignum, Virginia},
    title = { {A Clearer View on Fairness: Visual and Formal Representation for Comparative Analysis} },
    booktitle = {14th Scandinavian Conference on Artificial Intelligence, SCAI 2024},
    year = {2024},
    month = {June},
    pages = {112--120},
    editor = {Westphal, Florian and Peretz-Andersson, Einav and Riveiro, Maria and Bach, Kerstin and Heintz, Fredrik},
    organization = {Swedish Artificial Intelligence Society},
    doi = {10.3384/ecp208013},
    url = {https://ecp.ep.liu.se/index.php/sais/article/view/1005/913}
}

@misc{Mendez-2023-Soda,
    title = {{Soda: An Object-Oriented Functional Language for Specifying Human-Centered Problems}},
    author = {Mendez, Julian Alfredo},
    year = {2023},
    eprint = {2310.01961},
    archivePrefix = {arXiv},
    doi = {10.48550/arXiv.2310.01961},
    primaryClass = {
    id='cs.PL' full_name='Programming Languages' is_active=True alt_name=None in_archive='cs' is_general=False description='Covers programming language semantics, language features, programming approaches (such as object-oriented programming, functional programming, logic programming). Also includes material on compilers oriented towards programming languages; other material on compilers may be more appropriate in Architecture (AR). Roughly includes material in ACM Subject Classes D.1 and D.3.'
      }
}

@inproceedings{Mendez.Kampik-2025-EUMAS,
    author = {Mendez, Julian Alfredo and Kampik, Timotheus},
    title = { {Can Proof Assistants Verify Multi-agent Systems?} },
    booktitle = {Multi-Agent Systems},
    year = {2025},
    month = {June},
    pages = {323--339},
    editor = {Collier, Rem and Ricci, Alessandro and Nallur, Vivek and Burattini, Samuele and Omicini, Andrea},
    publisher = {Springer Nature Switzerland},
    address = {Cham},
    doi = {10.1007/978-3-031-93930-3_19}
}

@inproceedings{Mendez.Kampik-2025-LNGAI,
    author = {Mendez, Julian Alfredo and Kampik, Timotheus},
    title = {{Specification, Application, and Operationalization of a Metamodel of Fairness}},
    editor = {Liao, Beishui and Rotolo, Antonino and van der Torre, Leendert and Yu, Liuwen},
    booktitle = {{Logics for New-Generation AI}},
    volume = {5},
    month = {December},
    pages = {163--180},
    year = {2025},
    isbn = {978-1-84890-495-8},
    url = {https://www.collegepublications.co.uk/LNGAI/?00005}
}

@article{brannstrom2026humanemotionverificationaction, 
    title = {Human Emotion Verification by Action Languages via Answer Set Programming}, 
    doi = {10.1017/S1471068426100416}, 
    journal = {Theory and Practice of Logic Programming}, 
    author = {Br\"{a}nnstr\"{o}m, Andreas and Nieves, Juan Carlos}, 
    year = {2026}, 
    pages = {1--58} 
}

@article{cirillo2010human,
    title = {Human-aware task planning: an application to mobile robots},
    author = {Cirillo, Marcello and Karlsson, Lars and Saffiotti, Alessandro},
    journal = {ACM Transactions on Intelligent Systems and Technology (TIST)},
    volume = {1},
    number = {2},
    pages = {1--26},
    year = {2010},
    publisher = {ACM New York, NY, USA}
}

@article{gelfond1998action,
    title = {Action Languages},
    author = {Gelfond, Michael and Lifschitz, Vladimir},
    journal = {Computer and Information Science},
    volume = {3},
    number = {16},
    year = {1998},
    publisher = {Citeseer}
}

@article{dworschak2008mbn,
    title = {{Modeling {B}iological {N}etworks by {A}ction {L}anguages via {A}nswer {S}et {P}rogramming}},
    author = {Dworschak, S. and Grell, S. and Nikiforova, V.J. and Schaub, T. and Selbig, J.},
    journal = {Constraints},
    volume = {13},
    number = {1},
    pages = {21--65},
    year = {2008},
    publisher = {Springer}
}

@article{engelmann2023rv4jaca,
    title = {{RV4JaCa---Towards Runtime Verification of Multi-Agent Systems and Robotic Applications}},
    author = {Engelmann, Debora C and Ferrando, Angelo and Panisson, Alison R and Ancona, Davide and Bordini, Rafael H and Mascardi, Viviana},
    journal = {Robotics},
    volume = {12},
    number = {2},
    pages = {49},
    year = {2023},
    publisher = {MDPI}
}

@inproceedings{ferrando2025vitamin,
    title = {VITAMIN: VerIficaTion of A MultI ageNt system},
    author = {Ferrando, Angelo and Malvone, Vadim},
    booktitle = {Proceedings of the 24th International Conference on Autonomous Agents and Multiagent Systems},
    pages = {3023--3025},
    year = {2025}
}

@inproceedings{harmelen2021modular,
    title = {Modular Design Patterns for Systems that Learn and Reason},
    author = {Harmelen, Frank van},
    booktitle = {IEEE/WIC/ACM International Conference on Web Intelligence and Intelligent Agent Technology},
    pages = {4--4},
    year = {2021}
}

@book{mishra2024scalable,
    title = {Scalable AI and Design Patterns: Design, Develop, and Deploy Scalable AI Solutions},
    author = {Mishra, Abhishek},
    year = {2024},
    publisher = {Springer Nature}
}

@inproceedings{rudra2025composable,
    title = {Composable AI Stack for Intelligent Agents: Modular Orchestration Using Context Routing, Memory, and Tools},
    author = {Rudra, Angshuman and Agrawal, Manan},
    booktitle = {2025 International Conference on Computer and Applications (ICCA)},
    pages = {1--6},
    year = {2025},
    organization = {IEEE}
}

@article{monteith2024artificial,
    title = {Artificial intelligence and increasing misinformation},
    author = {Monteith, Scott and Glenn, Tasha and Geddes, John R and Whybrow, Peter C and Achtyes, Eric and Bauer, Michael},
    journal = {The British Journal of Psychiatry},
    volume = {224},
    number = {2},
    pages = {33--35},
    year = {2024},
    publisher = {Cambridge University Press}
}

@article{park2024ai,
    title = {AI deception: A survey of examples, risks, and potential solutions},
    author = {Park, Peter S and Goldstein, Simon and O'Gara, Aidan and Chen, Michael and Hendrycks, Dan},
    journal = {Patterns},
    volume = {5},
    number = {5},
    year = {2024},
    publisher = {Elsevier}
}

@article{sanchez2019designing,
    title = {Designing emotional BDI agents: good practices and open questions},
    author = {S{\'a}nchez-L{\'o}pez, Yanet and Cerezo, Eva},
    journal = {The Knowledge Engineering Review},
    volume = {34},
    year = {2019},
    publisher = {Cambridge University Press}
}

@inproceedings{pereira2007formal,
    title = {Formal modelling of emotions in BDI agents},
    author = {Pereira, David and Oliveira, Eug{\'e}nio and Moreira, Nelma},
    booktitle = {International Workshop on Computational Logic in Multi-Agent Systems},
    pages = {62--81},
    year = {2007},
    organization = {Springer}
}

@inproceedings{korevcko2013some,
    title = {On some concepts of emotional engine for BDI agent system},
    author = {Kore{\v{c}}ko, {\v{S}}tefan and Herich, Tom{\'a}{\v{s}}},
    booktitle = {2013 IEEE 14th International Symposium on Computational Intelligence and Informatics (CINTI)},
    pages = {527--532},
    year = {2013},
    organization = {IEEE}
}

@article{lorini2011logic,
    title = {A logic for reasoning about counterfactual emotions},
    author = {Lorini, Emiliano and Schwarzentruber, Fran{\c{c}}ois},
    journal = {Artificial Intelligence},
    volume = {175},
    number = {3-4},
    pages = {814--847},
    year = {2011},
    publisher = {Elsevier}
}

@article{adam2009logical,
    title = {A logical formalization of the OCC theory of emotions},
    author = {Adam, Carole and Herzig, Andreas and Longin, Dominique},
    journal = {Synthese},
    volume = {168},
    pages = {201--248},
    year = {2009},
    publisher = {Springer}
}

@inproceedings{dastani2012logic,
    title = {A logic of emotions: from appraisal to coping},
    author = {Dastani, Mehdi and Lorini, Emiliano},
    booktitle = {11th International Conference on Autonomous Agents and Multiagent Systems (AAMAS 2012)},
    pages = {1133--1140},
    year = {2012},
    organization = {ACM: Association for Computing Machinery}
}

@article{bolander2011epistemic,
    title = {Epistemic planning for single-and multi-agent systems},
    author = {Bolander, Thomas and Andersen, Mikkel Birkegaard},
    journal = {Journal of Applied Non-Classical Logics},
    volume = {21},
    number = {1},
    pages = {9--34},
    year = {2011},
    publisher = {Taylor \& Francis}
}

@inproceedings{davila2021simple,
    title = {A simple framework for cognitive planning},
    author = {Davila, Jorge Luis Fernandez and Longin, Dominique and Lorini, Emiliano and Maris, Fr{\'e}d{\'e}ric},
    booktitle = {Proceedings of the AAAI Conference on Artificial Intelligence},
    volume = {35},
    number = {7},
    pages = {6331--6339},
    year = {2021}
}

@inproceedings{lorini2022cognitive,
    title = {Cognitive planning in motivational interviewing},
    author = {Lorini, Emiliano and Sabouret, Nicolas and Ravenet, Brian and Davila, Jorge Luis Fernandez and Clavel, C{\'e}line},
    booktitle = {14th International Conference on Agents and Artificial Intelligence (ICAART 2022)},
    pages = {1--11},
    year = {2022}
}

@article{steunebrink2012formal,
    title = {A formal model of emotion triggers: an approach for BDI agents},
    author = {Steunebrink, Bas R and Dastani, Mehdi and Meyer, John-Jules Ch},
    journal = {Synthese},
    volume = {185},
    pages = {83--129},
    year = {2012},
    publisher = {Springer}
}

@article{balduccini2010formalization,
    title = {Formalization of psychological knowledge in answer set programming and its application},
    author = {Balduccini, Marcello and Girotto, Sara},
    journal = {Theory and Practice of Logic Programming},
    volume = {10},
    number = {4-6},
    pages = {725--740},
    year = {2010},
    publisher = {Cambridge University Press}
}

@inproceedings{jones2009personality,
    title = {Personality, Emotions and Physiology in a BDI Agent Architecture: The PEP-BDI Model},
    author = {Jones, Haza{\"e}l and Saunier, Julien and Lourdeaux, Domitile},
    booktitle = {2009 IEEE/WIC/ACM International Joint Conference on Web Intelligence and Intelligent Agent Technology},
    volume = {2},
    pages = {263--266},
    year = {2009},
    organization = {IEEE}
}

@book{lindholm2013java,
    title = {The Java virtual machine specification},
    author = {Lindholm, Tim and Yellin, Frank and Bracha, Gilad and Buckley, Alex},
    year = {2013},
    publisher = {Addison-wesley}
}

@article{Fikes-1971-Strips,
    title = {Strips: A new approach to the application of theorem proving to problem solving},
    journal = {Artificial Intelligence},
    volume = {2},
    number = {3},
    pages = {189-208},
    year = {1971},
    issn = {0004-3702},
    doi = {10.1016/0004-3702(71)90010-5},
    url = {https://www.sciencedirect.com/science/article/pii/0004370271900105},
    author = {Richard E. Fikes and Nils J. Nilsson}
}

@article{Levesque-1997-Golog,
    title = {GOLOG: A logic programming language for dynamic domains},
    journal = {The Journal of Logic Programming},
    volume = {31},
    number = {1},
    pages = {59-83},
    year = {1997},
    issn = {0743-1066},
    doi = {10.1016/S0743-1066(96)00121-5},
    url = {https://www.sciencedirect.com/science/article/pii/S0743106696001215},
    author = {Hector J. Levesque and Raymond Reiter and Yves Lesp\'{e}rance and Fangzhen Lin and Richard B. Scherl}
}

@article{Eiter-2003-DVLK,
    title = {A logic programming approach to knowledge-state planning, II: The DLVK system},
    journal = {Artificial Intelligence},
    volume = {144},
    number = {1},
    pages = {157-211},
    year = {2003},
    issn = {0004-3702},
    doi = {10.1016/S0004-3702(02)00367-3},
    url = {https://www.sciencedirect.com/science/article/pii/S0004370202003673},
    author = {Thomas Eiter and Wolfgang Faber and Nicola Leone and Gerald Pfeifer and Axel Polleres}
}

@techreport{McDermott-1998-PDDL,
    title = {PDDL---the planning domain definition language},
    author = {McDermott, Drew and Ghallab, Malik and Howe, Adele and Knoblock, Craig and Ram, Ashwin and Veloso, Manuela and Weld, Daniel and Wilkins, David},
    year = {1998},
    institution = {Yale Center for Computational Vision and Control}
}

@article{wang2020fake,
    title = {Fake news or bad news? Toward an emotion-driven cognitive dissonance model of misinformation diffusion},
    author = {Wang, Rui and He, Yuan and Xu, Jing and Zhang, Hongzhong},
    journal = {Asian Journal of Communication},
    volume = {30},
    number = {5},
    pages = {317--342},
    year = {2020},
    publisher = {Taylor \& Francis}
}

@article{ecker2022psychological,
    title = {The psychological drivers of misinformation belief and its resistance to correction},
    author = {Ecker, Ullrich KH and Lewandowsky, Stephan and Cook, John and Schmid, Philipp and Fazio, Lisa K and Brashier, Nadia and Kendeou, Panayiota and Vraga, Emily K and Amazeen, Michelle A},
    journal = {Nature Reviews Psychology},
    volume = {1},
    number = {1},
    pages = {13--29},
    year = {2022},
    publisher = {Nature Publishing Group US New York}
}

    \appendix

    \section{Appendix}

    \begin{proposition-appendix}[Properties of Contravenes]
        \normalfont
        \label{proof:contravenes-forbids-to-cause}

        Consider non-empty finite sets of fluent values $I, O$, an action $a$, consider the rules $r_{1} = \langle \mathbf{forbids\text{-}to\text{-}cause}, I, \emptyset, O \rangle$, $r_{2} = \langle \mathbf{causes\text{-}if}, I, \{a\}, O \rangle$, and $r_{3} = \langle \mathbf{contravenes}, I, \{a\}, \emptyset \rangle$. Consider that $a$ is not inhibited for $r_{1}$ and $r_{2}$. If a transition $t$ is valid for $r_{1}$ and for $r_{2}$, then it is valid for $r_{3}$, and if $t$ is invalid for $r_{1}$ and for $r_{2}$, then it is invalid for $r_{3}$.
    \end{proposition-appendix}

    \begin{proof}
        Let $t = \langle t_{I}, t_{A}, t_{O} \rangle$. Let us recall the semantic interpretations for the rules given in Definition~$\ref{def:verification-function}$:

        \begin{itemize}
            \item $r_{1} : (I \subseteq t_{I}) \Longrightarrow (O \cap t_{O} = \emptyset) $
            \item $r_{2} : (I \subseteq t_{I}) \land (\lvert A \rvert = 1) \land (A \subseteq t_{A}) \land (A \cap B = \emptyset) \Longrightarrow (O \subseteq t_{O}) $
            \item $r_{3} : (I \subseteq t_{I}) \Longrightarrow (\lvert A \rvert = 1) \land (A \cap t_{A} = \emptyset) $
        \end{itemize}
        where $B$ is the set of inhibited actions.

        Let us assume that $t$ is valid for $r_{1}$ and for $r_{2}$.
        If $r_{1}$ is not active for $t$, it means that $(I \subseteq t_{I})$ does not hold (by $r_{1}$), and therefore $r_{2}$ and $r_{3}$ are also not active. Then, $t$ is valid for $r_{3}$.
        If $r_{1}$ is active for $t$ and $r_{2}$ is not active for $t$, it means that $(\lvert A \rvert = 1) \land (A \subseteq t_{A}) \land (A \cap B = \emptyset) $ does not hold. Since $a$ is not inhibited, and by rule construction, the size of $A$ is always 1, then $(A \subseteq t_{A})$ does not hold. Thus, $A \cap t_{A} = \emptyset$ and $t$ is valid for $r_{3}$.
        Notice that $r_{2}$ cannot be active if $t$ is valid. To see this, let us assume that $r_{2}$ is active. That would mean that $r_{1}$ is also active for $t$, and in particular that $(O \cap t_{O} = \emptyset)$ and $(O \subseteq t_{O}) $. Since we assume that $t$ is valid, this can only be satisfied if $O = \emptyset$, which is not possible because $O$ is non-empty. Then, $r_{2}$ cannot be active if $t$ is valid.

        Let us assume that $t$ is invalid for $r_{1}$ and for $r_{2}$. Then, both $r_{1}$ and $r_{2}$ need to be active, which means that $(A \subseteq t_{A}) $ holds. Thus, $(A \cap t_{A} = \emptyset)$ does not hold, and $t$ is invalid for $r_{3}$.
        $\qed$
    \end{proof}

\end{document}